\documentclass[conference]{IEEEtran}
\IEEEoverridecommandlockouts

\usepackage[linesnumbered]{algorithm2e}
\usepackage{graphicx}
\usepackage{textcomp}
\usepackage{enumitem}
\usepackage{tabularx}
\usepackage{url}
\usepackage{floatrow}
\usepackage{booktabs}
\usepackage{multirow}
\usepackage{amsmath,amssymb,amsfonts}
\usepackage[misc,geometry]{ifsym} 
\usepackage[nodisplayskipstretch]{setspace}

\usepackage{amsthm}

\newtheorem{definition}{Definition}
\newtheorem{assumption}{Assumption}
\newtheorem{lemma}{Lemma}

\newcommand{\sig}[1]{{\small\textsf{{#1}}}}
\newcolumntype{L}[1]{>{\raggedright\arraybackslash}p{#1}}
\newcolumntype{C}[1]{>{\centering\arraybackslash}p{#1}}
\newcolumntype{R}[1]{>{\raggedleft\arraybackslash}p{#1}}

\newcommand{\Comment}[1]{}

\begin{document}

\title{Towards Safety Verification of Direct Perception Neural Networks
}

\author{\IEEEauthorblockN{Chih-Hong Cheng\IEEEauthorrefmark{1}, Chung-Hao Huang\IEEEauthorrefmark{2}, Thomas Brunner\IEEEauthorrefmark{2} and Vahid Hashemi\IEEEauthorrefmark{3}}
	\IEEEauthorblockA{
		\IEEEauthorrefmark{1}DENSO AUTOMOTIVE Deutschland GmbH}
	\IEEEauthorblockA{
		\IEEEauthorrefmark{2}fortiss - Research Institute of the Free State of Bavaria 
		}
	\IEEEauthorblockA{\IEEEauthorrefmark{3}Audi AG\\
		Contact: \texttt{c.cheng@denso-auto.de}
	}
}

\maketitle

\vspace{-5mm}

\begin{abstract}

We study the problem of safety verification of direct perception neural networks, where camera images are used as inputs to produce high-level features for autonomous vehicles to make control decisions. Formal verification of direct perception neural networks is extremely challenging, as it is difficult to formulate the specification that requires characterizing input as constraints, while the number of neurons in such a network can reach millions. We approach the specification  problem by learning an input property characterizer which carefully extends a direct perception neural network at close-to-output layers, and address the scalability problem by a novel assume-guarantee based verification approach. The presented workflow is used to understand a direct perception neural network (developed by Audi) which computes the next waypoint and orientation for autonomous vehicles to follow.

\end{abstract}


\section{Introduction}

Using deep neural networks has been the \emph{de facto} choice for developing visual object detection function in automated driving. Nevertheless, in the autonomous driving workflow, the neural networks can also be used more extensively. An example is \emph{direct perception}~\cite{chen2015deepdriving}; one trains a neural network to read high-dimensional inputs (such as images from camera or point clouds from lidar) and produce low-dimensional information called \emph{affordances} (e.g., safe maneuver regions or the next waypoint to follow) which could be used to program a controller for the autonomous vehicle. One may use direct perception as a hot standby system for a classical \emph{mediated perception} system that extracts objects and identifies lane markings before affordances are produced. 

In this paper, we study the \emph{safety verification problem} for a neural network implementing direct perception, where the goal is to ensure that under certain input conditions, the undesired output values never occur. An example of such a kind can be the following:
``\sig{For every input image where the road in the image strongly bends to the right, the output of the neural network  should never suggest to strongly steer to the left}''. Overall, the safety verification problem for direct perception networks is fundamentally challenging due to two factors:

\begin{itemize}
\item {\bf (Specification)} To perform safety verification, one premise is to have the undesired property formally specified. Nevertheless, it is practically impossible to characterize input specifications from images such as ``\sig{road strongly bends to the right}" and represent them as constraints over input variables. 

\item {\bf (Scalability)} Neural networks for direct perception often take images with millions of pixels, and the internal structure of the network can have many layers. This challenges any state-of-the-art formal analysis framework in terms of scalability. 
\end{itemize}

\begin{figure*}[t]
    \centering
    \includegraphics[width=0.85\textwidth, trim=0.7cm 1.2cm 2cm 0.7cm, clip]{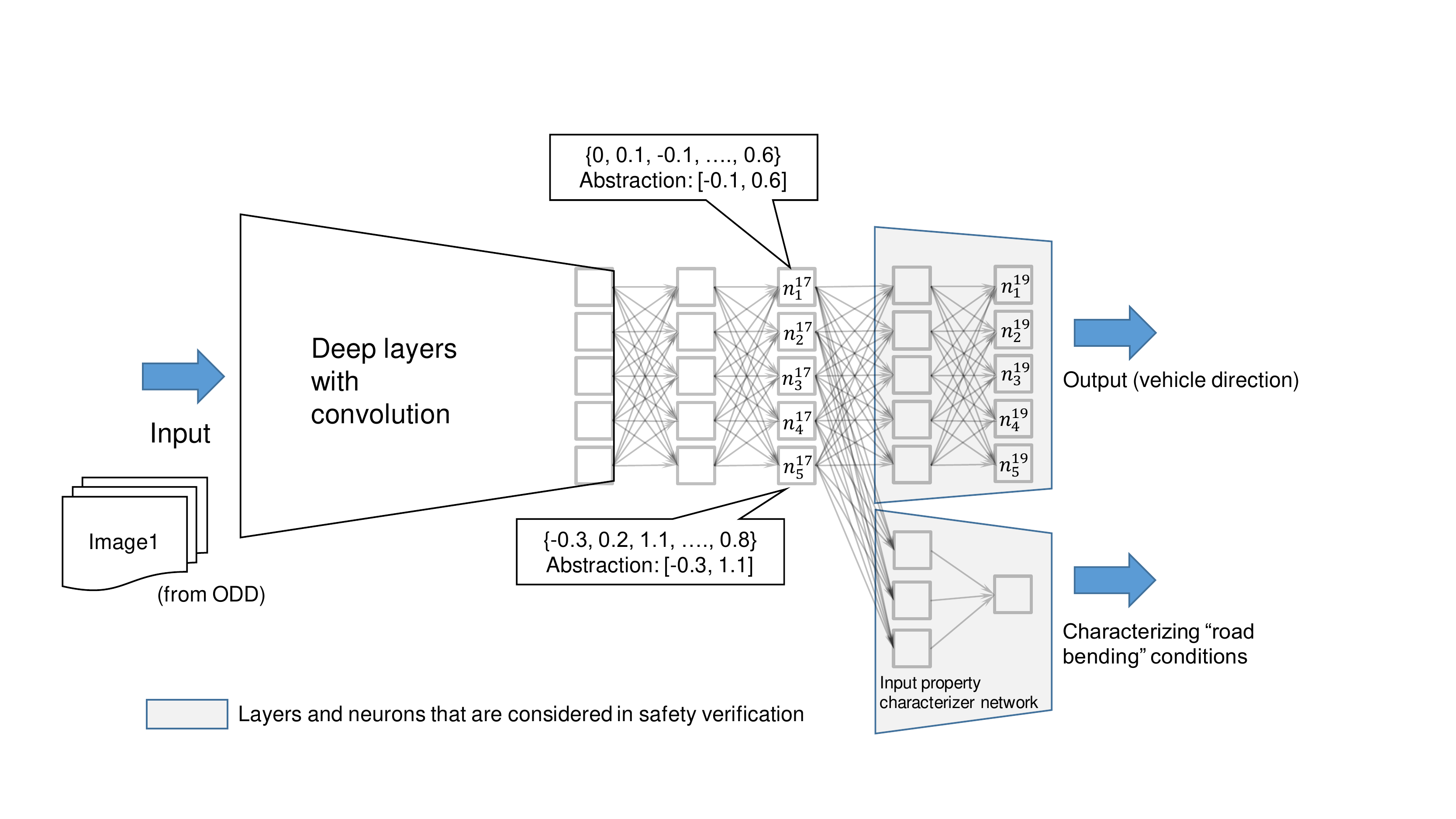}
    \caption{High-level illustration how to perform safety verification while tackling specification and scalability issues.}
    \label{fig:architecture}
\end{figure*}

Towards these issues, we present a workflow for safety verification of direct perception neural networks by simultaneously addressing  the specification and the scalability problem.  For the ease of understanding, we use Figure~\ref{fig:architecture} to explain the concept. First, we address the specification problem by learning an \emph{input property characterizer network}, where the input of the network is connected to close-to-output layer neurons of the original direct perception network. In Figure~\ref{fig:architecture}, the input property characterizer takes output values from the neurons $n^{17}_1$, $n^{17}_2$, $n^{17}_3$, $n^{17}_4$ and $n^{17}_5$ in the original deep perception network. For the previously mentioned specification, the input property characterizer outputs \sig{true} if an input image has ``\sig{road strongly bending to the right}". By doing so, the characterization of input features is aggregated to an output of a neural network. Subsequently, the safety verification problem is approached by asking if it is possible for the input-characterizing network to output \sig{true}, but the output of the direct perception network demonstrates  undesired values. As both the deep perception network and the input-characterizing network have shared neuron values, safety verification can be approached by only verifying close-to-output layers without losing  soundness. In Figure~\ref{fig:architecture}, safety verification only analyzes the sub-network colored grayed, and examines if any assignment of $n^{17}_1$, $n^{17}_2$, $n^{17}_3$, $n^{17}_4$, and $n^{17}_5$ leads to undesired output. The bounds of the neurons $n^{17}_1$, $n^{17}_2$, $n^{17}_3$, $n^{17}_4$, and $n^{17}_5$ can be decided by static analysis (which guarantees an overly conservative bound). However, using such a bound allows to have input images that are not possible to be seen in the operating design domain (ODD)\footnote{In training neural networks, the value for each pixel in an image is commonly re-scaled such that the re-scaled value is in the interval~$[0,1]$.  Starting verification using an input domain of~$[0,1]^{d_{l_0}}$ with~$d_{l_0}$ being the number of input image pixels, the result of formal verification always creates counter-examples in formal verification where counter-example images are so distant from what can be observed in practice (such as images without textures) and are rejected by experts.}. Thus we are advocating an alternative \emph{assume-guarantee} based approach where one first creates an outer polyhedron by aggregating all visited neuron values computed by the training set. We use the created polyhedron as a starting point to perform formal verification, by \emph{assuming} that for every possible input data in the ODD, the computed  neuron activation pattern is contained in this polyhedron. The assumption thus requires to be \emph{monitored in runtime} by checking if any computed neuron value falls outside the polyhedron. As an example, we consider the bound of~$n^{17}_1$ to be used in verification in Figure~\ref{fig:architecture}. By observing the minimum and the maximum of all visited values $\{0, 0.1, -0.1, \ldots, 0.6\}$,  $[-0.1, 0.6]$ is an over-approximation over all visited values, and one shall monitor in runtime whether the computed value of~$n^{17}_1$ has fallen outside~$[-0.1, 0.6]$.\footnote{Note that such a monitoring is needed regardless of formal verification, as neurons at close-to-output layers represent high-level features, so an image in operation that leads to unexpectedly high or low neuron feature intensity (indicated by falling outside the monitored interval) can be hints for incomplete data collection or indicators for the system stepping out from the ODD.}

The rest of the paper is organized as follows. Section~\ref{sec.verification.workflow} presents the required definitions as well as the workflow for verification. Section~\ref{sec.quantitative} discusses extensions to a statistical setup when the input property characterizer is not perfect. Lastly, we summarize related work in Section~\ref{sec.related} and conclude with our preliminary evaluation in Section~\ref{sec.conclusion}. 


\section{Verification Workflow}~\label{sec.verification.workflow}


A deep neural network is comprised of $L$ layers
where operationally,  the $l$-th layer for $l\in\{1,\dots,L\}$ of the network is a function $g^{(l)}: \mathbb{R}^{d_{l-1}} \rightarrow \mathbb{R}^{d_{l}}$, with $d_{l}$ being the dimension of layer~$l$.  
Given an input $\sig{in} \in \mathbb{R}^{d_{0}}$, the output of the $l$-th layer of the neural network $f^{(l)}$ is given by the functional composition of the $l$-th layer and the previous layers $f^{(l)}(\sig{in}) := \circ_{i=1}^{(l)} g^{(i)}(\sig{in})  = g^{(l)}(g^{(l-1)}\ldots g^{(2)}(g^{(1)}(\sig{in})))$.

\subsection{Characterizing Input Specification from Examples}

Let $\sig{In}_{\phi} \subseteq \mathbb{R}^{d_{0}}$ be the set of inputs of a neural network that satisfies the property~$\phi$. 
We assume that both~$\phi$ and~$\sig{In}_{\phi}$ are unknown (e.g., the road is bending left in an image), but there exists an oracle (e.g., human) that can answer for a given input~$\sig{in} \in \mathbb{R}^{d_{0}}$, whether $\sig{in} \in \sig{In}_{\phi}$. 

Let $(\sig{In}, \sig{C}_{\phi})$ be the list of training data and their associated labels (generated by the oracle) related to the input property~${\phi}$, where for every $(\sig{in}, \sig{c}) \in (\sig{In}, \sig{C}_{\phi})$, $\sig{in} \in \mathbb{R}^{d_{0}}$, $\sig{c} \in \{\sig{0}, \sig{1}\}$, we have $(\sig{in}, 1 ) \in (\sig{In}, \sig{C}_{\phi})$ \emph{if} $\sig{in} \in \sig{In}_{\phi}$ and $(\sig{in}, 0 ) \in (\sig{In}, \sig{C}_{\phi})$ \emph{if} $\sig{in} \not\in \sig{In}_{\phi}$. The \emph{perfect input property  characterizer} extending the $l$-th layer is a function~$h_{l}^{\phi}$ which guarantees that for every $(\sig{in}, \sig{c}) \in (\sig{In}, \sig{C}_{\phi})$, $h_{l}^{\phi} (f^{(l)}(\sig{in})) = \sig{c}$. The generation of~$h_{l}^{\phi}$ can be done by training a neural network as a binary classifier, with~$100\%$ success rate on the training data. The following assumption states that as long as function~$h_{l}^{\phi}$ performs perfectly on the training data, $h_{l}^{\phi}$ will also perfectly generalize to the complete input space. In other words, we can use $h_{l}^{\phi}$ to characterize~$\phi$.

\begin{assumption}[Perfect  Generalization]~\label{assumption.perfect.qualitative}
Assume that $h_{l}^{\phi}$ also perfectly characterizes $\phi$, i.e., $\forall \sig{in}\in \mathbb{R}^{d_{0}}$: $h_{l}^{\phi}(f^{(l)}(\sig{in})) = \sig{1} $ iff $ \sig{in} \in \sig{In}_{\phi}$.
\end{assumption}

\begin{definition}[Safety Verification]
The safety verification problem asks if there exists an input~$\sig{in} \in \sig{In}_{\phi}$ such that~$f^{(L)}(\sig{in})$ satisfies~$\psi$, where the \textbf{risk} condition~$\psi$ is a conjunction of linear inequalities over the output of the neural network. If no such input~$\sig{in}$ exists, we say that the neural network is safe under the input constraint~$\phi$ and the output risk constraint~$\psi$. 

\end{definition}

When Assumption~\ref{assumption.perfect.qualitative} holds, for safety verification it is \emph{equivalent} to ask whether there exists an input~$\sig{in} \in \mathbb{R}^{d_{0}}$ such that~$h_{l}^{\phi}(f^{(l)}(\sig{in})) = \sig{1}$ and~$f^{(L)}(\sig{in})$ satisfies~$\psi$. From now on, unless explicitly specified, we consider only situations where  Assumption~\ref{assumption.perfect.qualitative} holds.

\subsection{Practical Safety Verification}

\paragraph{Abstraction by omitting neurons before the $l$-th layer.} The following result states that one can retain soundness for safety verification, by considering all possible neuron values that can appear in the~$l$-th layer. 

\begin{lemma}[Verification by Layer Abstraction]~\label{lemma.abstraction.layer}
If there exists no $\hat{\sig{n}_l}\in \mathbb{R}^{d_l}$ such that 
$g^{(L)}( g^{(L-1)}\ldots(g^{(l+1)}(\hat{\sig{n}_l}))$ satisfies~$\psi_{out}$ and $h_{l}^{\phi}(\hat{\sig{n}_l}) =\sig{1}$, then the neural network is safe under input constraint~$\phi$ and output risk constraint~$\psi$.
\end{lemma}

\begin{proof} The lemma holds because for every input~$\sig{in} \in \mathbb{R}^{d_{0}}$ of the network, $f^{(l)}(\sig{in}) \in  \mathbb{R}^{d_l}$. 
\end{proof}


\noindent Obviously, the use of~$\mathbb{R}^{d_l}$ in Lemma~\ref{lemma.abstraction.layer} is overly conservative, and we can strengthen Lemma~\ref{lemma.abstraction.layer} without losing soundness, if we find $\mathcal{S} \subseteq \mathbb{R}^{d_{l}}$ which guarantees that $f^{(l)}(\sig{in}) \in \mathcal{S}$ for every input~$\sig{in} \in \mathbb{R}^{d_{0}}$ of the network. Obtaining such a set~$\mathcal{S}$ can be achieved by abstract interpretation techniques~\cite{gehr2018ai2,yang2019analyzing} which perform symbolic reasoning over the neural network in a layer-wise manner. 

\begin{lemma}[Abstraction via Input Over-approximation]~\label{lemma.abstraction.layer.v2}
Let $\mathcal{S} \subseteq \mathbb{R}^{d_{l}}$ guarantee that $f^{(l)}(\sig{in}) \in \mathcal{S}$ for every input~$\sig{in} \in \mathbb{R}^{d_{0}}$ of the network. 
If there exists no $\hat{\sig{n}_l}\in \mathcal{S}$ such that 
$g^{(L)}( g^{(L-1)}\ldots(g^{(l+1)}(\hat{\sig{n}_l}))$ satisfies~$\psi_{out}$ and $h_{l}^{\phi}(\hat{\sig{n}_l}) =\sig{1}$, then the neural network is safe under input constraint~$\phi$ and output risk constraint~$\psi$.
\end{lemma}

\paragraph{Assume-guarantee Verification via Monitoring.} If the computed~$\mathcal{S}$, due to over-approximation, is too coarse to prove safety, one practical alternative is to generate $\tilde{\mathcal{S}}$ which only guarantees $f^{(l)}(\sig{in}) \in \tilde{\mathcal{S}}$ for every input~$\sig{in} \in \sig{In}$ in the training data. In other words, $\tilde{\mathcal{S}}$ over-approximates the neuron values computed based on the samples in the training data. 

If using $\tilde{\mathcal{S}}$ is sufficient to prove safety and if  for any input $\sig{in}$, checking whether  $f^{(l)}(\sig{in}) \in \tilde{\mathcal{S}}$ can be computed efficiently, one can conditionally accept the proof by designing a run-time monitor which raises a \emph{warning} that the  assumption  $f^{(l)}(\sig{in}) \in \tilde{\mathcal{S}}$ used in the proof is violated.  Admittedly, $\tilde{\mathcal{S}}$ can be an under-approximation over $\{f^{(l)}(\sig{in}) \;|\; \sig{in} \in \mathbb{R}^{d_{0}}\}$, but practically creating an over-approximation only based on the training data is useful and can avoid unstructured input such as noise which is allowed when using~$\mathbb{R}^{d_{0}}$.

\section{Towards Statistical Reasoning}~\label{sec.quantitative}

The results in Section~\ref{sec.verification.workflow} are based on two assumptions of \emph{perfection}, namely

\begin{itemize}
    \item {\bf(perfect training)} the input property characterizer perfectly decides whether property~$\phi$ holds, for each sample in the training data, and
    \item {\bf(perfect generalization)} the input property characterizer generalizes its decision (whether property~$\phi$ holds) also perfectly to every data point in the complete input space.  
\end{itemize}

One important question appears when the above two assumptions do not hold, meaning that it is possible for the input property characterizer to make mistakes. By considering all four possibilities in Table~\ref{table.four.possibilities}, one realizes that even when a safety proof is established by considering all inputs where $h_{l}^{\phi}(f^{(l)}(\sig{in})) = \sig{1}$, there exists a probability~$\gamma$ where an input~$\sig{in}$ should be analyzed, but $\sig{in}$ is omitted in the proof process due to $h_{l}^{\phi}(f^{(l)}(\sig{in}))$ being $\sig{0}$ (i.e., $\sig{in} \in \sig{In}_{\phi}$ and $h_{l}^{\phi}(f^{(l)}(\sig{in})) = \sig{0}$). Therefore, one can only establish a statistical guarantee with\\$(1-\gamma)$ probability over the correctness claim\footnote{Note that for parts where $\sig{in} \not\in \sig{In}_{\phi}$ and $h_{l}^{\phi}(f^{(l)}(\sig{in})) = \sig{0}$, no problem occurs as the safety analysis guarantees the desired property when $h_{l}^{\phi}(f^{(l)}(\sig{in})) = \sig{1}$.}, provided that all data points used in training $h^{\phi}_{l}$ are also safe\footnote{In other words, for every $(\sig{in}, \sig{c}) \in (\sig{In}, \sig{C}_{\phi})$, if $h_{l}^{\phi}(f^{(l)}(\sig{in})) = \sig{0}$ and $\sig{c} = \sig{1}$, then
$f^{(L)}(\sig{in})$ does not satisfy~$\psi$.}.

\begin{table}[t]
\begin{tabular}{C{2.5cm}|C{2cm} | C{2cm}|}
\cline{2-3}
                        & $\sig{in} \in \sig{In}_{\phi}$ & $\sig{in} \not\in \sig{In}_{\phi}$ \\ \hline
\multicolumn{1}{|l|}{$h_{l}^{\phi}(f^{(l)}(\sig{in})) = \sig{1}$} & $\alpha$ & $\beta$ \\ \hline
\multicolumn{1}{|l|}{$h_{l}^{\phi}(f^{(l)}(\sig{in})) = \sig{0}$} & $\gamma$ & $1 - \alpha -\beta -\gamma$ \\ \hline
\end{tabular}
\caption{Probability by considering all possible cases due to decisions made by the input characterizer (whether $h_{l}^{\phi}(f^{(l)}(\sig{in})) = \sig{1}$) and the ground truth (whether $\sig{in} \in \sig{In}_{\phi}$).}\label{table.four.possibilities}
\vspace{-2mm}
\end{table}

\section{Related Work}~\label{sec.related}

Formal verification of neural networks has drawn huge attention with many results available~\cite{pulina2010abstraction,DBLP:conf/cav/KatzBDJK17,cheng2017maximum,ehlers2017formal,lomuscio2017approach,narodytska2018verifying,gehr2018ai2,dutta2018output,bunel2018unified,DBLP:conf/ijcai/RuanHK18,weng2018towards,wang2018formal,yang2019analyzing}. Although specifications used in formal verification of neural networks are discussed in recent reviews~\cite{huang2018safety,seshia2018formal}, the specification problem in terms of characterising an image set is not addressed, so research results largely use inherent properties of a neural network such as local robustness (as output invariance) or output ranges where one does not need to characterize properties over a set of input images. Not being able to properly characterizing input conditions (one possibility is to simply consider every input to be bounded by $[-1, 1]$) makes it difficult for formal static analysis to achieve any useful results on deep perception networks, regardless of the type of abstraction domain being used (box, octagon, or zonotope). Lastly, our work is motivated by \emph{zero shot learning}~\cite{palatucci2009zero} which trains additional features apart from a standard neural network. The feature detector is commonly created by extending the network from close-to-output layers.

\section{Evaluation and Concluding Remarks}~\label{sec.conclusion}

We have applied this methodology to examine a direct perception neural network developed by Audi. The network acts as a hot standby system and computes the next waypoint and orientation for autonomous vehicles to follow. As the close-to-output layers of the network are either ReLU or Batch Normalization, and as $\psi$ is a conjunction of linear constraints over output, it is feasible to use exact verification methods such as ReLUplex~\cite{DBLP:conf/cav/KatzBDJK17}, Planet~\cite{ehlers2017formal} or MILP-based approaches~\cite{cheng2017maximum, lomuscio2017approach} as the underlying verification method. We developed a variation of \sig{nn-dependability-kit}\footnote{\url{https://github.com/dependable-ai/nn-dependability-kit/}} to read models from TensorFlow\footnote{\url{https://www.tensorflow.org/}} and to perform formal verification via a reduction to MILP. Using assume-guarantee based techniques that take an over-approximation from neuron values produced by the training data\footnote{The data is taken from a particular segment of the German A9 highway, by considering variations such as weather and the current lane.}, it is possible to conditionally prove some properties such as ``\emph{impossibility to suggest steering to the far left, when the road image is bending to the right}". However, under the current setup, it is still impossible to prove intriguing properties such as  ``\emph{impossibility to suggest steering straight, when the road image is bending to the right}". We suspect that the main reason is due to the inherent limitation of the neural network under analysis.

In our experiment, we also found that for some input properties such as traffic participants in adjacent lanes, it is very difficult to construct the corresponding input property characterizers by taking neuron values from close-to-output layers (i.e., the trained classifier almost acts like fair coin flipping). Based on the theory of \emph{information bottleneck} for neural networks~\cite{tishby2015deep,shwartz2017opening}, a neural network from high dimensional input to low dimensional output naturally eliminates unrelated information in close-to-output layers. Therefore, the input property can be unrelated to the output of the network. Although we are unable to prove that the output of the network is safe under these input constraints, it should be possible to construct a counter example either by capturing more data or by using adversarial perturbation techniques~\cite{szegedy2013intriguing,moosavi2017universal}.   

To achieve meaningful formal verification, in our experiments, we also realized that it is commonly not sufficient to only record the minimum and maximum value for each neuron, as boxed abstraction can lead to huge over-approximation. In certain circumstances, we also record the minimum and maximum difference between two adjacent neurons in a layer (in Figure~\ref{fig:architecture}, we record $n^{17}_{i+1} - n^{17}_{i}$ where $i \in \{1,2,3,4\}$). Modern training frameworks such as TensorFlow support  computing  differences of adjacent neurons with GPU parallelization\footnote{Computing the neuron difference, when neuron values are stored in an 1D tensor~$n$ can be done in numpy using a single instruction $\texttt{diff}(n)$, and in TensorFlow using $n[1:]-n[:-1]$.}, thereby making monitoring possible.

Overall, our initial result demonstrates the potential of using formal methods even on very complex neural networks, while it provides a clear path to engineers to resolve the problem related to how to characterize input conditions for verification (by also applying machine learning techniques). Our approach of looking at close-to-output layers can be viewed as an abstraction which can, in future work, leads to layer-wise incremental abstraction-refinement techniques. Although our practical motivation is to verify direct perception networks, the presented technique is equally  applicable to any deep network for vision and lidar systems where input constraints are hard to characterize. It opens a new research direction of using learning to assist practical verification of learning systems.

\vspace{2mm}
\noindent \textbf{Acknowledgement} The research work 
is conducted during the first author's service at the fortiss research institute and
is supported by the following projects: ``\emph{Audi Verifiable AI}" from Audi AG, Germany and  ``\emph{Dependable AI for automotive systems}" from DENSO Corporation, Japan. 

\bibliographystyle{abbrv}

\end{document}